\documentclass[sigconf]{acmart}

\copyrightyear{2024}
\acmYear{2024}
\setcopyright{rightsretained}
\acmConference[KDD '24]{Proceedings of the 30th ACM SIGKDD Conference on
Knowledge Discovery and Data Mining}{August 25--29, 2024}{Barcelona, Spain}
\acmBooktitle{Proceedings of the 30th ACM SIGKDD Conference on Knowledge
Discovery and Data Mining (KDD '24), August 25--29, 2024, Barcelona,
Spain}
\acmDOI{10.1145/3637528.3671757}
\acmISBN{979-8-4007-0490-1/24/08}

\makeatletter
\gdef\@copyrightpermission{
  \begin{minipage}{0.3\columnwidth}
   \href{https://creativecommons.org/licenses/by/4.0/}{\includegraphics[width=0.90\textwidth]{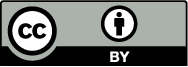}}
  \end{minipage}\hfill
  \begin{minipage}{0.7\columnwidth}
   \href{https://creativecommons.org/licenses/by/4.0/}{This work is licensed under a Creative Commons Attribution International 4.0 License.}
  \end{minipage}
  \vspace{5pt}
}
\makeatother

\usepackage{hyperref}
\usepackage{amsmath}
\usepackage{soul}
\usepackage{graphicx}
\usepackage{algorithm}
\usepackage[noend]{algpseudocode}
\usepackage{natbib}
\usepackage{xspace}
\usepackage{subfig}
\usepackage{bbm}
\usepackage{bm}
\usepackage{centernot}
\usepackage{tcolorbox}
\usepackage{balance}

\usepackage{tikz}
\usepackage{pgfplots}

\definecolor{yafcolor1}{rgb}{0.4, 0.165, 0.553}
\definecolor{yafcolor2}{rgb}{0.949, 0.482, 0.216}
\definecolor{yafcolor3}{rgb}{0.47, 0.549, 0.306}
\definecolor{yafcolor4}{rgb}{0.925, 0.165, 0.224}
\definecolor{yafcolor5}{rgb}{0.141, 0.345, 0.643}
\definecolor{yafcolor6}{rgb}{0.965, 0.933, 0.267}
\definecolor{yafcolor7}{rgb}{0.627, 0.118, 0.165}
\definecolor{yafcolor8}{rgb}{0.878, 0.475, 0.686}

\newcommand{\dmax}{d_{\text{max}}}
\newcommand{\dmin}{d_{\text{min}}}
\newcommand{\TC}{\text{TC}}

\usepackage{algorithm}
\usepackage[noend]{algpseudocode}

\newcommand{\bigO}{\ensuremath{\mathcal{O}}\xspace}
\newcommand{\NP}{\ensuremath{\mathbf{NP}}\xspace}

\algrenewcommand\algorithmicrequire{\textbf{Input:}}
\algrenewcommand\algorithmicensure{\textbf{Output:}}

\usepackage{array}
\usepackage{arydshln}
\newcommand{\ptitle}[1]{\smallskip\noindent{\bf #1.}}
\newcommand{\pttitle}[1]{\smallskip\noindent{\it #1.}}

\newcommand{\set}[1]{\left\{#1\right\}}

\newcommand{\fpr}[1]{\mathopen{}\left(#1\right)}

\newcommand{\abs}[1]{{\left|#1\right|}}
\newcommand{\floor}[1]{{\left\lfloor#1\right\rfloor}}

\newcommand{\define}{\leftarrow}

\DeclareRobustCommand{\dispfunc}[2]{%
    \ensuremath{%
        \ifthenelse{\equal{#2}{}}%
            {\mathit{#1}}%
            {\mathit{#1}\fpr{#2}}}}

\newcommand{\diver}[1]{\dispfunc{div}{#1}}

\newcommand{\algclust}[1]{\dispfunc{\textsc{Clust}}{#1}}
\newcommand{\algextract}[1]{\dispfunc{\textsc{Extract}}{#1}}
\newcommand{\algbac}[1]{\dispfunc{\textsc{BAC}}{#1}}
\newcommand{\algbacb}[1]{\dispfunc{\textsc{BCR}}{#1}}
\newcommand{\algbacf}[1]{\dispfunc{\textsc{BCF}}{#1}}

\newtheorem{theorem}{Theorem}
\newtheorem{lemma}{Lemma}
\newtheorem{corollary}{Corollary}

\newtheorem{prop}{Proposition}
\newtheorem{prob}{Problem}

\pgfdeclarelayer{background}
\pgfdeclarelayer{foreground}
\pgfsetlayers{background,main,foreground}

\definecolor{yafaxiscolor}{rgb}{0.3, 0.3, 0.3}

\newlength{\yafaxispad}
\newlength{\yaftlpad}
\newlength{\yaflabelpad}
\newlength{\yafaxiswidth}
\newlength{\yafticklen}
\makeatletter
\def\pgfplots@drawtickgridlines@INSTALLCLIP@onorientedsurf#1{}
\makeatother

\newcommand{\yafdrawaxis}[4]{
	\pgfplotstransformcoordinatex{#1}\let\xmincoord=\pgfmathresult 
	\pgfplotstransformcoordinatex{#2}\let\xmaxcoord=\pgfmathresult 
	\pgfplotstransformcoordinatey{#3}\let\ymincoord=\pgfmathresult 
	\pgfplotstransformcoordinatey{#4}\let\ymaxcoord=\pgfmathresult 
	\pgfsetlinewidth{\yafaxiswidth} 
	\pgfsetcolor{yafaxiscolor}
	\pgfpathmoveto{\pgfpointadd{\pgfpointadd{\pgfplotspointrelaxisxy{0}{0}}{\pgfqpointxy{\xmincoord}{0}}}{\pgfqpoint{-0.5\yafaxiswidth}{\yafaxispad}}}
	\pgfpathlineto{\pgfpointadd{\pgfpointadd{\pgfplotspointrelaxisxy{0}{0}}{\pgfqpointxy{\xmaxcoord}{0}}}{\pgfqpoint{0.5\yafaxiswidth}{\yafaxispad}}}
	\pgfpathmoveto{\pgfpointadd{\pgfpointadd{\pgfplotspointrelaxisxy{0}{0}}{\pgfqpointxy{0}{\ymincoord}}}{\pgfqpoint{\yafaxispad}{-0.5\yafaxiswidth}}}
	\pgfpathlineto{\pgfpointadd{\pgfpointadd{\pgfplotspointrelaxisxy{0}{0}}{\pgfqpointxy{0}{\ymaxcoord}}}{\pgfqpoint{\yafaxispad}{0.5\yafaxiswidth}}}
	\pgfusepath{stroke}
}

\pgfplotscreateplotcyclelist{yaf}{%
{yafcolor5,mark options={scale=0.75},mark=o}, 
{yafcolor2,mark options={scale=0.75},mark=square},
{yafcolor3,mark options={scale=0.75},mark=triangle},
{yafcolor4,mark options={scale=0.75},mark=o},
{yafcolor1,mark options={scale=0.75},mark=o},
{yafcolor8,mark options={scale=0.75},mark=o},
{yafcolor6,mark options={scale=0.75},mark=o},
{yafcolor7,mark options={scale=0.75},mark=o}} 

\pgfkeys{/pgf/number format/.cd,1000 sep={\,}}
\pgfplotsset{axis y line=left, axis x line=bottom,
	tick align=outside,
	tickwidth=\yafticklen,
	clip = false,
    x axis line style= {-, line width = 0pt, color=black!0},
    y axis line style= {-, line width = 0pt, color=black!0},
    x tick style= {line width = \yafaxiswidth, color=yafaxiscolor, yshift = \yafaxispad},
    y tick style= {line width = \yafaxiswidth, color=yafaxiscolor, xshift = \yafaxispad},
    x tick label style = {font=\small, yshift = \yaftlpad, inner xsep = 0pt},
    y tick label style = {font=\small, xshift = \yaftlpad},
    every axis y label/.style = {at = {(ticklabel cs:0.5)}, rotate=90, anchor=center, font=\small, yshift = -\yaflabelpad, inner sep = 0pt},
    every axis x label/.style = {at = {(ticklabel cs:0.5)}, anchor=center, font=\small, yshift = \yaflabelpad},
    x tick label style = {font=\small, yshift = 1pt},
    grid = major,
    major grid style  = {dash pattern = on 1pt off 3 pt},
	every axis plot post/.append style= {line width=\yafaxiswidth} ,
	legend cell align = left,
	legend style = {inner sep = 1pt, cells = {font=\scriptsize}},
	legend image code/.code={%
		\draw[mark repeat=2,mark phase=2,#1] 
		plot coordinates { (0cm,0cm) (0.15cm,0cm) (0.3cm,0cm) };%
	} 
}

\begin{document}

\title{Max-Min Diversification with Asymmetric Distances}


\author{Iiro Kumpulainen}
\email{iiro.kumpulainen@helsinki.fi}
\affiliation{%
  \institution{University of Helsinki}
  \city{Helsinki}
  \country{Finland}
}

\author{Florian Adriaens}
\email{florian.adriaens@helsinki.fi}
\affiliation{%
  \institution{University of Helsinki, HIIT}
  \city{Helsinki}
  \country{Finland}
}

\author{Nikolaj Tatti}
\email{nikolaj.tatti@helsinki.fi}
\affiliation{%
  \institution{University of Helsinki, HIIT}
  \city{Helsinki}
  \country{Finland}
}

\renewcommand{\shortauthors}{Iiro Kumpulainen, Florian Adriaens, \& Nikolaj Tatti}

\begin{abstract}
One of the most well-known and simplest models for diversity maximization is the Max-Min Diversification (MMD) model, which has been extensively studied in the data mining and database literature.
In this paper, we initiate the study of the Asymmetric Max-Min Diversification (AMMD) problem. The input is a positive integer $k$ and a complete digraph over $n$ vertices, together with a nonnegative distance function over the edges obeying the directed triangle inequality. The objective is to select a set of $k$ vertices, which maximizes the smallest pairwise distance between them.
AMMD reduces to the well-studied MMD problem in case the distances are symmetric, and has natural applications to query result diversification, web search, and facility location problems.
Although the MMD problem admits a simple $\frac{1}{2}$-approximation by greedily selecting the next-furthest point, this strategy fails for AMMD and it remained unclear how to design good approximation algorithms for AMMD. 

We propose a combinatorial $\frac{1}{6k}$-approximation algorithm for AMMD by leveraging connections with the Maximum Antichain problem. We discuss several ways of speeding up the algorithm and compare its performance against heuristic baselines on real-life and synthetic datasets.
\end{abstract}

\maketitle

\section{Introduction}
\label{sec:intro}
Diversity maximization is a fundamental problem
with natural applications to query result diversification \cite{drosou2012disc,deng2013complexity}, recommender systems \cite{castells2021novelty,kaminskas2016diversity, abbassi2013diversity}, information exposure in social networks \cite{matakos2020tell}, web search \cite{xin2006extracting, radlinski2006improving,divtopk,bhattacharya2011consideration}, feature selection \cite{zadeh2017scalable} and data summarization \cite{celis2018fair,mahabadi2023improved,CHANDRA2001438,zheng2017survey}.
In a typical diverse data selection problem, one is interested in finding a \emph{diverse} group of items within the data, meaning that the selected items should be highly dissimilar to each other.
Typically, one is interested in finding sets of fixed and small sizes.

\begin{figure}
\centering
\includegraphics{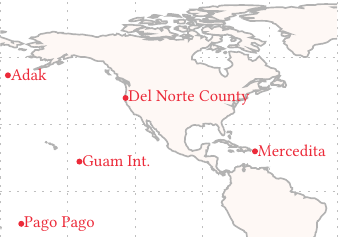}

\caption{The result of our algorithm \algbacf{} for $k=5$ on the \emph{Flights Delay} dataset \cite{FlightsDelay}. It shows 5 airports located in U.S. territory which require a large flight time between any two of them, averaged over all flights in 2015. The airports are all spread out over different U.S. territories (Puerto Rico, Guam, American Samoa) and/or states (Alaska, California).
\label{fig:introflightex}}

\end{figure}

A simple and popular model for diversity maximization is the Max-Min Diversification (MMD) model. In this model, the data is represented as a universe of points $U$ of size $|U|=n$, and we are given a nonnegative distance function $d: U \times U \to \mathbb{R}_{\geq 0}$ measuring the dissimilarity between points. If $d(u,v)$ is large, then $u$ and $v$ are highly dissimilar and vice versa. 
The space $(U,d)$ is a pseudometric\footnote{Contrary to a metric, a pseudometric allows $d(u,v)=0$ for distinct $u \neq v$.}, implying that for all $u,v,w \in U: d(u,u) = 0$, $d(u,v) = d(v,u)$ (symmetry), and $d(u,v) \leq d(u,w)+ d(w,v)$ (triangle inequality).
The \emph{diversity score} $\diver{S}$ of a set of points $S \subseteq U$ is then defined as the smallest distance between any two distinct points in $S$. Formally,
\begin{align}
\label{def:divscore}
\diver{S} = \min \limits_{\substack{u,v \in S \\ u \neq v}} d(u,v).
\end{align}
The Max-Min Diversification (MMD) problem asks, given an integer $k$, to find a set $S \subseteq U$ of size $|S| = k$ which maximizes $\diver{S}$.
The MMD problem has originally been studied in the operations research literature as the max-min dispersion problem or $k$-dispersion problem \cite{kuby1987programming,erkut1990discrete,tamir1991obnoxious,ravi1994heuristic}.
The problem naturally describes a facility location problem where the proximity of selected facilities is undesirable.
Facilities that are close to each other might be undesirable due to economic or safety reasons.
Examples include the placement of store warehouses, nuclear power plants, or ammunition depot storages.

This paper initiates the study of the \emph{asymmetric generalization} of the MMD problem, which will be referred to as the Asymmetric Max-Min Diversification problem (AMMD). AMMD allows for asymmetric distances $d(u,v) \neq d(v,u)$ for $u \neq v \in U$, while retaining all other properties of the MMD model such as the triangle inequality.
A formal problem statement is given by Problem~\ref{prob:ammd}.

There are two good reasons for studying AMMD.
First, it can be argued that many real-world similarity measures and distances are in fact not symmetric. Straightforward examples are one-way traffic roads, uphill/downhill traveling, one-sided friendships and asymmetric flight times between airports, among others~\cite{kunegis2013konect, snapnets, tsplib}. Hence, one should consider studying models incorporating this property.

Secondly, for many optimization problems, the asymmetric generalization requires different, often more complicated, approximation algorithms than their symmetric restrictions. 
Additionally, they might also be significantly harder to approximate. An example is the $\Omega(\log^* n)$-hardness result \cite{chuzhoy2005asymmetric} for asymmetric $k$-center,\!\footnote{The iterated logarithm $\log^* n$ is defined as the number of times the $\log$ function can be iteratively applied to $n$ until the result is less than or equal to $1$.} and the corresponding $\bigO(\log^* n)$-approximation algorithms \cite{archer2001two,panigrahy1998ano}. Another example is the asymmetric traveling salesman problem, for which only recently a constant-factor approximation has been discovered \cite{svensson2020constant}.

Designing approximation algorithms for AMMD is seemingly also more challenging than for MMD.
While MMD admits a greedy $\frac{1}{2}$-approximation by iteratively selecting the next-furthest point until a set of size $k$ is obtained \cite{tamir1991obnoxious,ravi1994heuristic},
this greedy strategy can perform arbitrarily bad on asymmetric instances. Figure~\ref{fig:toyexample_greedy} shows such an example.

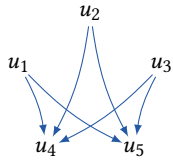
\begin{figure}
  \centering
  \begin{tikzpicture}
  \node[inner sep=0pt, circle] (u2) at (90:1cm) {$u_2$};
  \node[inner sep=0pt, circle] (u1) at (162:1cm) {$u_1$};
  \node[inner sep=0pt, circle] (u4) at (234:1cm) {$u_4$};
  \node[inner sep=0pt, circle] (u5) at (306:1cm) {$u_5$};
  \node[inner sep=0pt, circle] (u3) at (378:1cm) {$u_3$};

  \draw[->, >=latex, yafcolor5] (u1) edge[bend left = 10] (u4);
  \draw[->, >=latex, yafcolor5] (u1) edge[bend right = 10] (u5);

  \draw[->, >=latex, yafcolor5] (u2) edge[bend left = 10] (u4);
  \draw[->, >=latex, yafcolor5] (u2) edge[bend right = 10] (u5);

  \draw[->, >=latex, yafcolor5] (u3) edge[bend left = 10] (u4);
  \draw[->, >=latex, yafcolor5] (u3) edge[bend right = 10] (u5);
  \end{tikzpicture}
  \caption{An input graph to AMMD on which the known $\frac{1}{2}$-approximations for symmetric MMD perform poorly. A blue directed edge $(u,v)$ indicates that the distance $d(u,v)$ is zero. All other distances (edges not drawn) are equal to some $R>0$. These distances satisfy the directed triangle inequality. For $k=3$ the optimal solution is $O =\{u_1,u_2,u_3\}$, with optimum $\diver{O}=R$. The greedy algorithm from \cite{tamir1991obnoxious} picks an arbitrary initial node. If $u_4$ or $u_5$ are selected we get a solution value of zero, regardless of how the remaining two nodes are selected. Similarly, the algorithm of \cite{ravi1994heuristic} initially selects a node pair with maximum distance. This could be $(u_4,u_5)$, since ties are broken arbitrarily, again resulting in a zero-valued solution.}
\label{fig:toyexample_greedy}
\end{figure} 

\ptitle{Results and techniques}
Throughout the paper we let $\omega < 2.373$ be the matrix multiplication exponent. We refer to \cite{alman2021refined} for the current best bound.

Our main result is the following theorem.
\begin{theorem}
Our algorithms \algbac{}, \algbacb{} and \algbacf{} from Section~\ref{sec:speeding} approximate AMMD within a factor of $\frac{1}{6k}$. Their worst-case time complexity is respectively $\bigO(n^{2+\omega} \log k)$, $\bigO(n^{2+\omega} \log k \log
n)$ and $\bigO(n^{\omega} \log k \log n )$.
\end{theorem} 

To the best of our knowledge, these are the first approximation algorithms for AMMD with a non-trivial guarantee. \algbac{} is the vanilla algorithm. \algbacb{} is a refined variant looking to improve solution quality, while \algbacf{} is a faster variant.

We also show that our practical implementation run fast on real-world and synthetic data and produces solutions with high diversity scores (see Section~\ref{sec:exps}). Figure~\ref{fig:introflightex} shows the output of \algbacf{} for $k=5$ on a real-world dataset consisting of U.S. domestic flight data by large air carriers in 2015 \cite{FlightsDelay}. Between the five airports, the shortest average flight time is from Adak Airport to Del Norte Country Airport. Flying from Mercedita Airport to Guam Int. Airport on average took 10\% longer than flying in the reverse direction. This is the largest ratio of bidirectional flying times among all routes between the five airports.

The high-level idea behind our algorithm is to create auxiliary graphs based on the distances in the AMMD instance. In these auxiliary graphs, an independent set of size $k$ will lead to good solutions. In order to efficiently extract the independent sets, we need to equip the auxiliary graphs with certain properties. We show that by first clustering the points in the AMMD instance, the graphs must contain a subgraph of a certain type (antichain, a large chordless cycle, or a long shortest path with no backward edges). We can search for all 3 subgraphs in polynomial time, and we can extract an independent set from them in polynomial time.

\ptitle{Roadmap} Section~\ref{sec:related} discusses related work on diversity maximization and the MMD problem in particular. Section~\ref{sec:notation} provides a formal problem statement and the necessary background literature on maximum antichains.
Section~\ref{sec:algosnearlysym} briefly discusses the performance of the greedy algorithm on instances that are nearly symmetric. Section~\ref{sec:approx} details our algorithms for AMMD. We present experimental evaluation in Section~\ref{sec:exps} and conclude the paper with remarks in Section~\ref{sec:conclusions}.

\section{Related Work}
\label{sec:related}
\ptitle{MMD} The MMD problem was shown to be \NP -complete by \citet{wang1988study}.
Both \citet{tamir1991obnoxious} and \citet{ravi1994heuristic} showed that the greedy next-furthest point algorithms yield a $\frac{1}{2}$-approximation for MMD. \citet{moumoulidou2020diverse} noted that these greedy algorithms are essentially identical to the well-known Gonzalez-heuristic \cite{gonzalez1985clustering} which gives a 2-approximation for the $k$-center problem. By a reduction from max. clique, \citet{ravi1994heuristic} showed that for any $\epsilon>0$ MMD is \NP -hard to approximate within a factor of $\frac{1}{2}+\epsilon$, thereby showing the optimality of the Gonzalez-heuristic in terms of worst-case performance.
\citet{ravi1994heuristic} also showed that MMD without the triangle inequality constraints is \NP -hard to approximate within any multiplicative factor.

\citet{akagi2018exact} showed that MMD can be solved exactly by solving $\bigO(\log n)$ $k$-clique problems, resulting in an exact algorithm in $\bigO(n^{\omega k/3}\log n)$ time by using the $k$-clique algorithm from \cite{nevsetvril1985complexity}.

\ptitle{Fairness and matroid constraints} Fairness variants of the MMD problem have recently gained a lot of research attention \cite{Addanki0MM22,celis2018fair,wang2023max,Moumoulidou0M21,wang2022streaming,wang2023fair}.
A typical fair MMD problem statement assumes the universe is partitioned into $|C|$ disjoint groups. Solutions are now required to contain a certain amount of points from each group. It is still undetermined whether constant-factor approximation algorithms exist for arbitrary $|C|$ \cite{wang2023max,Addanki0MM22}.
The fairness constraints are a special case of diversity maximization under matroid constraints \cite{abbassi2013diversity,borodin2012max,ceccarello2020general,ceccarello2018fast,cevallos2017local}.

\ptitle{Other objective functions} \citet{CHANDRA2001438} considered a range of diversity problems with different objective functions besides MMD. Among them is the popular Max-Sum Diversification (MSD) problem, which aims to maximize the sum of pairwise distances instead of the minimum \cite{kuby1987programming,abbassi2013diversity,aghamolaei2015diversity,bauckhage2020adiabatic,
borassi2019better, indyk2014composable, hassin1997approximation,bhattacharya2011consideration}. 
\citet{cevallos2017local} provided a local search $(1-\bigO(1/k))$-approximation for MSD for distances of negative type (which includes several non-metric distances), subject to general matroid constraints. \citet{zhang2020maximizing} consider a variant of MSD adapted to a clustered data setting. Besides MSD there are numerous other objective functions for diversity maximization, depending on the topic area. We refer to~\cite{zheng2017survey} for a summary of results in query result diversification.

Nonetheless, all of these works consider symmetric distances between points. We also note that the asymmetric generalization of the MSD problem is trivially approximated by considering the symmetric distance $d'$ formed by $d'(u,v) = d(u,v)+d(v,u)$ and applying any MSD approximation algorithm on $d'$. This strategy does not lead to a guarantee for AMMD, as the distance in one direction might be very small compared to the other direction.

\section{Notation and Preliminaries}
\label{sec:notation}
\ptitle{Graphs}
Given a directed graph (digraph) $G=(V, A)$, vertex $v$ is reachable from vertex $u$ if there exists a path from $u$ to $v$ in $G$. Otherwise, $v$ is called unreachable from $u$. The transitive closure $\TC(G)$ of a digraph $G$ is a digraph with the same nodes as $G$, and with edges $(u,v) \in \TC(G)$ if $v$ is reachable from $u$ in $G$.
We say a digraph $G$ is transitively closed if $G=\TC(G)$.
A set $X \subseteq V$ is called independent if there are no edges in the induced subgraph $G[X]$.
An edge $(i,j) \in A$ is often abbreviated as $ij$.

\ptitle{AMMD problem}
As mentioned in the introduction, we will assume a pseudometric space $(U,d)$ without symmetry constraints on the distance function $d$, thus satisfying $d(u,u)=0$, nonnegativity $d(u,v) \geq 0$ and the directed triangle inequality $d(u,v) \leq d(u,w)+d(w,v)$ for all $u,v,w \in U$. Instances of AMMD can be viewed as complete digraphs or distance matrices representing the distances between each pair of points. An example of an AMMD instance space is the
\emph{metric closure} of a nonnegatively weighted digraph, formed by defining distances as weighted shortest path lengths between all the pairs of nodes in the graph.

We are interested in the following problem.
\begin{prob}[AMMD]
\label{prob:ammd}
Given $(U, d)$ and integer $k$, find a set $O \subseteq U$ with $\abs{O} = k$ such that $\diver{O}$, as defined in Eq.~\ref{def:divscore}, is maximized.
\end{prob}
Throughout the paper, we let $R^* = \diver{O}$ be the optimal value for an AMMD instance space $(U,d)$ with parameter $k$, where $O$ is a corresponding optimal set of $k$ points. We assume $R^* > 0$, as when $R^*=0$ any set of $k$ points would be an optimal solution.


With $d$ we define two new functions $\dmin$ and $\dmax$ as $\dmin(u,v) = \min\{d(u,v), d(v,u)\}$ and $\dmax(u,v) = \max\{d(u,v), d(v,u)\}$ for all $u,v \in U$. Both functions are clearly symmetric, and $\dmax$ satisfies the triangle inequality while $\dmin$ does not. Finally, using the $\dmin$ and $\dmax $ functions we quantify the asymmetry of a given pseudometric space. A pseudometric space $(U, d)$ is called \emph{$\epsilon$-symmetric} for some $\epsilon \geq 0$ if for all $u,v \in U: (1+\epsilon)\dmin(u,v) \geq \dmax(u,v)$.

\ptitle{Maximum antichains}
\label{sec:maxanti}
We will see that there is a connection between AMMD and the maximum independent set (MIS) problem.
Unfortunately, finding a MIS is not feasible. Instead, we will look for a maximum antichain (MA).
An \emph{antichain} of a digraph $G=(V, E)$ is a set of vertices that are pairwise unreachable in $G$. In other words, an antichain is an independent set in $\TC(G)$,
and therefore also in $G$.
We have the following problem.
\begin{prob}[MA]\label{prob:maxa}
Given a digraph $G$ and integer $k$, find an antichain of size $k$ or decide no such set exists.
\end{prob}

Unlike finding a MIS, MA can be solved in polynomial time.
A common approach is to reduce the problem into a minimum flow (with edge demands) problem~\cite{ntafos1979path,marchal2018parallel,
caceres2023minimum,pijls2013another}. Such reduction has $2|V|+ 2$ vertices and $3|V|+|E|$ edges. The minimum flow problem is then solved by reducing it to a maximum flow problem, we refer to \cite{caceres2023minimum} for the explicit construction.

This reduction has been combined with the recent breakthrough paper on maximum flows by~\citet{chen2022maximum}
to obtain the following result.
\begin{prop}[Corollary 2.1 \cite{caceres2022minimum}\label{caceres}]
Given a digraph $G=(V,E)$ with $|V|=n$ and $|E|=m$, a maximum antichain of $G$ can be computed in $\bigO(m^{1 + o(1)} \log n)$ time with high probability.
\end{prop}

There is a subtle complication in using the solver by~\citet{chen2022maximum}.
The algorithm solves the maximum flow problem with high probability, that is, there is a small $o(1)$ chance of failure.\!\footnote{By definition, this probability goes to 0 as the instance size grows.}
The issue is that in certain cases we need to solve MA several times, say $s$ times, and \emph{all} of them need to succeed. 
Here, the asymptotic probability of failure depends on how quickly the failure rate goes to 0 for a single MA as the instance size grows.
However, if we make the failure rate of a single instance $o(s^{-1})$, then the union bound immediately shows that the probability of a single failure
in $s$ instances is then in $o(1)$. The standard technique for guaranteeing an $o(s^{-1})$ failure bound is to run the solver $\bigO(\log s)$ times
and select the best result.
To summarize, we have the following proposition.
\begin{prop}\label{prop:multima}
Given $s$ digraphs $G_i = (V_i,E_i)$ with $\abs{V_i} \leq n$ and $\abs{E_i} \leq m$, computing $s$ maximum antichains of each $G_i$ can be done in $\bigO(s m^{1 + o(1)} \log n \log s)$ time with high probability.
\end{prop}

We should point out that in practice we do not use the solver by~\citet{chen2022maximum} due to its high complexity and high constants.
Instead, we settle for a more practical algorithm by~\citet{goldberg2008partial} that solves maximum flow instances in $\bigO(n^2\sqrt{m})$ time. 

We conclude this section by discussing alternative techniques for solving maximum antichains.
The first approach
requires the computation of the transitive closure of the digraph. The idea is to reduce the MA problem to finding a maximum cardinality matching in a bipartite graph encoding the reachability relation between vertices~\cite{felsner2003recognition, dilworth1987decomposition, fulkerson1956note}.
The downside of this approach is that computing the transitive closure is expensive.

Moreover, there is a large body of work on width-parametrized\footnote{The size of the maximum antichain is often called the \emph{width} of the graph.} algorithms for the MA problems~\cite{caceres2023minimum,caceres2022minimum,caceressoda,kowaluk2008path,makinen2019sparse,felsner2003recognition}. These algorithms are practical in a small-width regime~\cite{caceres2023minimum}.
However, we opted to use the maximum flow approach as this was sufficiently fast for us, while also being the asymptotically fastest algorithm.
\section{Nearly symmetric instances}
\label{sec:algosnearlysym}
Greedily selecting the next-furthest points until $k$ points are selected is a $\frac{1}{2}$-approximation for symmetric MMD \cite{tamir1991obnoxious,ravi1994heuristic}, but can perform arbitrarily badly
on asymmetric instances as shown by the example in Figure~\ref{fig:toyexample_greedy}.

Theorem~\ref{thm:epsisymmgreedy} generalizes this result to asymmetric instances.
It states that on asymmetric instances that are $\epsilon$-symmetric (see Section~\ref{sec:notation}), the greedy approach applied on the $\dmin$ distances yields a $\frac{1}{2+\epsilon}$-approximation, and this ratio is tight.
The proof of Theorem~\ref{thm:epsisymmgreedy} can be found in the Appendix.

\begin{theorem}
\label{thm:epsisymmgreedy}
For any $\epsilon \geq 0$, Algorithm~\ref{algo:greedydmin} is a $\frac{1}{2+\epsilon}$-approximation on $\epsilon$-symmetric instances and can be implemented to run in $\bigO(kn)$ time.
Additionally, there exist $\epsilon$-symmetric instances for which Algorithm~\ref{algo:greedydmin} cannot achieve a performance ratio better than $\frac{1}{2+\epsilon}$.  
\end{theorem}

\begin{algorithm}[t]
\caption{Greedy with $\dmin$-distances.}
\label{algo:greedydmin}
\begin{algorithmic}[1]
\Require space $(U, d)$ and integer parameter $k\geq 2$.
\State $v \define \text{arbitrary vertex from } U$.
\State $S \define \{v\}$.
\While {$|S|<k$}
	\State $v \define \text{arg} \max_{u \in U} \dmin (u,S)$.
	\State $S \define S \cup \{v\}$.
\EndWhile
\Ensure the set $S$. 
\end{algorithmic}
\end{algorithm}

\section{Ball-and-antichain method}
\label{sec:approx}
Section~\ref{sec:approxnk2} details a straightforward approximation algorithm for AMMD, exploiting the polynomial time complexity of the MA problem in digraphs, as discussed in the previous section.
This algorithm has an approximation guarantee of $\frac{1}{n-k+1}$, which is not very useful in a typical regime of small $k$, but it gives insight into the use of the MA problem for approximating AMMD.

In Section~\ref{sec:approx16k} we modify the algorithm from Section~\ref{sec:approxnk2} by first clustering the points based on the $\dmax$ distances between them. The subspace induced by the cluster centers has some very useful properties, which leads to an approximation algorithm with a multiplicative guarantee of $\frac{1}{6k}$ for AMMD.

Finally, in Section~\ref{sec:speeding}, we discuss improvements to look for better solutions and speed up the algorithm while maintaining the approximation guarantee.

\begin{algorithm}[t]
\caption{Naive Maximum Antichain method.}
\label{algo:app_ma}
\begin{algorithmic}[1]
\Require space $(U, d)$ and integer parameter $k\geq 2$.
\ForAll {$R \in \{d(i,j)>0 \mid i, j \in U,  i \neq j\}$}
	\State Create $G_R = (U,A)$, with $ij \in A \Leftrightarrow d(i,j)<\frac{R}{n-k+1}$. \label{line:naive_create_G}
	\State $M \define$ \text{Maximum antichain of} $G_R$.
	\If{$|M| \geq k$,} $S_R \leftarrow \text{any }k \text{ points from } M$.

	\EndIf
\EndFor
\Ensure the set $S_R$ with the largest $\diver{S_R}$ value. 
\end{algorithmic}
\end{algorithm}

\subsection{Naive approach based on antichains}
\label{sec:approxnk2}

We begin by describing the naive approach for approximating AMMD.
Assume for the moment that we know $R = R^*$, the optimal value for an AMMD instance.
Consider a digraph $G = (U, A)$, where $ij \in A$ if and only if $d(i, j) < R$.
Then an independent set, say $O$, in $G$ of size $k$ will have $\diver{O} = R^*$.
Unfortunately, finding a maximum independent set in a graph is an \NP-hard problem with a weak approximation guarantee~\cite{hastad1996clique}.

Therefore, we lower the cutoff by setting it to $\frac{R}{n-k+1}$. This makes the underlying graph
so sparse that we can guarantee that the graph contains an antichain, say $S$, of size $k$.
Since an antichain is also an independent set, we know that $\diver{S} \geq \frac{R}{n-k+1}$.
We can find the antichain in polynomial time. Finally, we do not know $R^*$ but we know that it
is one of the distances. Therefore, we test every distance; there are at most $n(n - 1)$ of such distances.
The pseudo-code for the algorithm is given in Algorithm~\ref{algo:app_ma}.

\begin{theorem}
\label{thm:naivemaxanti}
Algorithm~\ref{algo:app_ma} is an $\frac{1}{n-k+1}$-approximation for AMMD in $\bigO(n^{4 + o(1)} \log n)$ time with high probability.
\end{theorem}

The proof is given in Appendix.

We finalize this section by observing that we could binary search for the largest $R$ value for which $G_R$ (defined in line~\ref{line:naive_create_G} in Algorithm~\ref{algo:app_ma}) still has an antichain of size $k$. First, sort the unique distances in time $\bigO(n^2 \log n)$, after which we need at most $\bigO(\log n)$ calls to find this $R$. Since for this $R$ we have $R \geq R^*$, we retain the same approximation guarantee. The binary search performs $\bigO(\log n)$ MA computations, and all of them need to succeed. Proposition~\ref{prop:multima} implies that the algorithm solves the problem in
$\bigO(n^{2 + o(1)}\log^2 n \log \log n)$ time with high probability.

\subsection{Refined approach: clustering and antichains}
\label{sec:approx16k}

The problem with Algorithm~\ref{algo:app_ma} is that it is using a very conservative cutoff of $\frac{R}{n-k+1}$, leading
to a weak guarantee. 
We show that we can relax this cutoff to $\frac{R}{6k}$, by first clustering the space $(U,d)$ according to the $\dmax$ distances.
The discovered cluster centers will have the property that two centers must have a large $\dmax$ distance between them.
A consequence of this property is that the resulting graph contains a large antichain, a large chordless cycle, or a large shortest path with no backward edges.
It turns out that we can search for all 3 subgraphs in polynomial time, and using those we can extract an independent set in polynomial time.

\ptitle{Clustering step}
Next, we describe the clustering step, as given in
Algorithm~\ref{algo:ballcover}. Here the algorithm greedily covers $U$ with a family of pairwise disjoint sets $\set{A_t}$. 
Each $A_t$ is constructed by selecting an unmarked point as a center $c_t$, and adding all unmarked points $v$ with $\dmax(c_t, v) < R$. Since all the vertices in $A_t$ are marked at the end of step $t$ (line~5), they cannot be selected by any $A_{t'}$ for $t' > t$. It follows that for every $t \neq t'$ we have $A_t \cap A_{t'} = \emptyset$. 

Algorithm~\ref{algo:ballcover} terminates in at most $n$ steps, since every set $A_t$ contains at least one point, namely the center point $c_t$. Algorithm~\ref{algo:ballcover} runs in $\bigO(n^2)$ time.

Each $A_t$ is constructed by selecting vertices that have a small $\dmax$ distance to their center $c_t$. Line~4 in Algorithm~\ref{algo:ballcover} ensures that the $\dmax$ distances between two distinct centers are at least $R$. 

To analyze this further, let us define $R'$ as the smallest distance that is at least one third of the optimum $R^*$,
\begin{equation}
\label{def:R'}
R' = \min\{d(u,v) \mid d(u,v) \geq R^*/3, \ u,v \in U, \ u\neq v\}.
\end{equation}
If we then perform the clustering for any $R \leq R'$, there will be $k$ centers that also have a pairwise $\dmin$ distance of at least $R^*/3$ between them. This is captured by Proposition~\ref{prop:cluster} and Corollary~\ref{cor:clusterphase}.

\begin{algorithm}[t]
\caption{$\algclust{U, d, R}$, clusters $U$ according to $\dmax$.}\label{algo:ballcover} 
\begin{algorithmic}[1]
\Require space $(U, d)$ and parameter $R > 0$.
\State Label all $u \in U$ as unmarked, let ${U}' \define \emptyset$ and  $t \leftarrow 1$.
\While {there exists an unmarked point}
	\State $c_t \define $ any unmarked point.
	\State $A_t \define \{\text{unmarked } v \in U \mid \dmax(c_t,v) < R \}.$ 
	\State Mark all $v \in A_t$. 
	\State ${U}' \define {U}' \cup \{c_t\}$ and $t \define t+1$.
\EndWhile
\Ensure ${U}'.$
\end{algorithmic}
\end{algorithm}

\begin{prop}
\label{prop:cluster}
Let $O$ be an optimal solution to AMMD with optimum $R^*$ and $R'$ as defined in Equation~\ref{def:R'}.
If $R \leq R'$, then the following two statements regarding Algorithm~\ref{algo:ballcover} are true.
\begin{itemize}
\item For all $t$ it holds that $|A_t \cap O| \leq 1$.
\item For any $t\neq t'$ for which $|A_t \cap O|=1$ and $|A_{t'} \cap O|=1$, it holds that $\dmin(c_t,c_{t'}) \geq R' \geq R^*/3$.
\end{itemize}
\end{prop}
\begin{proof}
For the first statement, suppose that $A_t$ contains two distinct $x, y \in O, x \neq y$.
Then it holds that $d(x,y) \leq d(x,c_t)+d(c_t,y)$. Since $d(x,c_t)$ and $d(c_t,y)$ are strictly less than $R\leq R'$, they must be less than $R^*/3$ because $R'$ is defined as the smallest distance greater or equal to $R^*/3$, so any distance strictly smaller than $R'$ must be less than $R^*/3$. Thus, $d(x,y) < 2R^*/3 < R^*$, a contradiction since $x, y \in O, x \neq y$ implies that $d(x,y) \geq R^*$.

For the second statement, assume that $A_t$ contains $x \in O$ and $A_{t'}$ contains $y \in O$. Since $A_t$ and $A_{t'}$ are disjoint, we have $x \neq y$. Note that $d(x,y) \leq d(x,c_t)+d(c_t,c_{t'})+d(c_{t'},y)$, where $d(x,c_t)<R \leq R'$ and $d(c_{t'},y)<R \leq R'$. This means $d(x,c_t)$ and $d(c_{t'},y)$ must be less than $R^*/3$ by the definition of $R'$. So if $d(c_t,c_{t'}) < R'$, we would have $d(x,y) < 3 R^*/3 = R^*$, a contradiction. Similarly, we cannot have $d(c_{t'},c_t) < R'$ either, which means $\dmin(c_t,c_{t'}) \geq R' \geq R^*/3$.
\end{proof}

\begin{corollary}
\label{cor:clusterphase}
Let ${U}' = \algclust{U, d, R}$.
For every $u, v \in {U}', u \neq v$ we have $\dmax(u,v) \geq R$. Additionally, if $R \leq R'$, then ${U}'$ contains a set $S$, for which $|S| = k$ and $\diver{S} \geq R' \geq R^*/3$.

\end{corollary}
 
Corollary~\ref{cor:clusterphase} states that as long as $R \leq R'$, from any instance space $(U, d)$ we can efficiently find a subset ${U}' \subseteq U$ 
such that ${U}'$ still contains $k$ points with a pairwise $\dmin$ distance of at least $R^*/3$ between them. This enables us to restrict ourselves to ${U}'$, at the expense of a decrease in the optimal value by a factor of three.

\ptitle{The \algbac{} algorithm}
We are ready to describe our algorithm which we call \algbac{} (shortened for ball-and-antichain). The pseudocode is given in Algorithms~\ref{algo:bac}--\ref{algo:extract}.
Similar to the naive approach we iterate over all distances. For each candidate distance $R$,
we cluster the space to get the centers $U'$.
We then construct a graph $G$ with edges corresponding to distances shorter than $\frac{R}{2k}$.
We can guarantee that there is ($i$) a large chordless cycle, ($ii$) a long shortest path with no backward edges,
or ($iii$) a large antichain. In the first two cases, we can obtain an independent set by selecting $k$ vertices with odd indices.
In the last case, it is enough to select $k$ vertices from the found antichain.

\begin{algorithm}[t]
\caption{\algbac{U, d, k}, an $\frac{1}{6k}$-approx. algorithm for AMMD.}\label{algo:bac} 
\begin{algorithmic}[1]
\Require space $(U, d)$ and integer parameter $k\geq 2$.
\ForAll {$R \in \{d(i,j)>0 \mid i, j \in U,  i \neq j\}$} \label{line:loop}
	\State ${U}' \define \algclust{U, d, R}$. 
	\State $\algextract{U', d, R/(2k), k}$.
\EndFor 

\Ensure the set returned by $\algextract{}$ with the largest $\diver{}$ value.
\end{algorithmic}
\end{algorithm}

\begin{algorithm}[t]
\caption{\algextract{U, d, \sigma, k}, subroutine for extracting a candidate set.}\label{algo:extract} 
\begin{algorithmic}[1]
\Require space $(U, d)$, threshold $\sigma$, and integer parameter $k\geq 2$.
    \State Create $G = (U,A)$, with $ij \in A \Leftrightarrow d(i,j)< \sigma$.\label{line:create_G}
    \If{$G$ contains a cycle}
        \State $C \define$ chordless cycle in $G$. \label{line:cycle}
    \EndIf
    \State $G_c \define$ the condensation of $G$.
    \State $M \define$ maximum antichain of $G_c$.
    \State $L \define$ shortest path of length $2k-1$ in $G_c$ or longest found. \label{line:path}
    \If{$C$ exists and $|C| \geq 2|M|-1$ and $|C| \geq |L|$}
        \State $I \define$ points with odd indices from $C$.
    \ElsIf{$2|M|-1 \geq |L|$} 
        \State $I \define$ points in $G$ corresponding to points in $M$.
    \Else{} 
        \State $I \define$ points in $G$ corr. to points with odd indices in $L$.
    \EndIf
\Ensure greedily selected $k$ points from the set $I$, if found.
\end{algorithmic}
\end{algorithm}

Next, we will prove the approximation guarantee. 
First, we need Lemma~\ref{lem:cycles}, which states that there cannot exist small cycles in $G$.

\begin{lemma}
\label{lem:cycles}
For any $R > 0$, any cycle $C$ in the digraph $G$ constructed in $\algextract{U', d, \frac{R}{2k}, k}$ (see Alg.~\ref{algo:extract})
has at least $2k + 2$ distinct vertices.
\end{lemma}
\begin{proof}
Suppose $G$ contains a cycle $C = (v_1,\ldots,v_{\ell},v_{1})$ of length $\ell$. 
Since $C$ is a cycle in $G$, it holds that $d(v_{\ell},v_1) < \frac{R}{2k}$, by definition of $G$.
On the other hand, as ${U}'$ is the output of \algclust{}, Corollary~\ref{cor:clusterphase} states that $\dmax(v_1,v_{\ell}) \geq R$. This implies that $d(v_1,v_{\ell}) \geq R$. Now the triangle inequality for $d(v_1,v_{\ell})$ along the edges of cycle $C$ implies
\[
	R \leq d(v_1,v_{\ell}) \leq \sum_{i = 1}^{\ell - 1} d(v_i, v_{i + 1}) < (\ell - 1) \frac{R}{2k}. 
\]
Solving for $\ell$ leads to $\ell > 2k + 1$, which proves the claim.
\end{proof}

\begin{theorem}
\label{thm:approx}
\algbac{U, d, k} is an $\frac{1}{6k}$-approximation to AMMD. 
\end{theorem}

\begin{proof}
Line~\ref{line:loop} in Algorithm~\ref{algo:bac} iterates over all unique distances, and one of them is equal to $R'$ as defined in Eq.~\ref{def:R'}.
We will show that for $R \leq R'$, the digraph $G$, constructed in line~\ref{line:create_G} in Algorithm~\ref{algo:extract}, has an antichain $M$ of size $|M| \geq k$, or there exists either a shortest path
with no backward edges or a chordless cycle from which we can select $k$ independent vertices. An independent set $I$ of size $\abs{I} \geq k$ in graph $G$ then yields a solution with a diversity score of $\diver{I} \geq \frac{R}{2k}$, which for $R = R'$ is $\diver{I} \geq\frac{R'}{2k} \geq \frac{R^*}{6k}$ proving the theorem. 

Note that since the nodes in an antichain have no paths connecting them, it suffices to look for an antichain in the condensation $G_c$, whose vertices are the strongly connected components of $G$. If there exists an antichain $M$ of size $|M| \geq k$, we are done as the nodes in an antichain are independent. 

Consider the case where the maximum antichain $M$ has size $|M|<k$ and assume $G$ is a DAG. Note that when $G$ is a DAG the condensation $G_c$ is equivalent to $G$.

Corollary~\ref{cor:clusterphase} states that if $R \leq R'$ then ${U}'$ contains a subset
$S \subseteq {U}'$ for which $|S| = k$ and $\diver{S} \geq R'$. Then there must be a path in $G$ between some pair of distinct points in $S$. Otherwise, $S$ is an antichain of size $k$.

Let this pair of nodes be $x, y \in S, x \neq y$, with a path $(x = v_1, \ldots, v_\ell = y)$ from $x$ to $y$ in $G$.
Then the triangle inequality implies
\[
	R' \leq \diver{S} \leq d(x, y) \leq \sum_{i=1}^{\ell-1} d(v_i, v_{i + 1}) < (\ell - 1) \frac{R}{2k} \leq (\ell - 1) \frac{R'}{2k}.
\]
Therefore, $\ell$ must be at least $2k + 2$ meaning any path between $x$ and $y$ must have at least $2k + 2$ vertices. Hence, there is a shortest path of length $2k + 2$ while a shortest path $L$ of length $2k - 1$ is sufficient.
There cannot be any shortcut edges in $L$, since $L$ is a shortest path. Nor can there be any backward edges in $L$, since $G$ is a DAG.
Consequently, elements in $L$ with odd indices form an independent set $I$ of size $k$.

Finally, assume that $G$ is not a DAG. Then
there is a chordless cycle $C$. Lemma~\ref{lem:cycles} guarantees that $C$ has at least $2k + 2$ elements.
Then, elements in $C$ with odd indices form an independent set $I$ of size $\abs{I} \geq k+1 > k$.
\end{proof}

\ptitle{Time complexity of Algorithm~\ref{algo:bac}}
The iteration in line~\ref{line:loop} is over at most $\bigO(n^2)$ possible $R$
values.
Both detecting a cycle in $G$ and extracting the chordless
cycle from it (line~\ref{line:cycle} of Algorithm~\ref{algo:extract}) take $\bigO(n^2)$ time. Computing the
maximum antichain can be done in $\bigO(n^{2 + o(1)} \log n)$ time
(Proposition~\ref{caceres}).

To compute the shortest path we can use the following approach:
Let $D = A + I$, where $A$ is the adjacency matrix of $G$,
and $I$ is the identity matrix. Then $D^\ell_{ij} > 0$ if and only if
there is a path of at most length $\ell$ from $i$ to $j$.
Consequently, there is a shortest path from $i$ to $j$ of length $2k - 1$ if and only if
$D^{2k-1}_{ij} > 0$ and $D^{2k-2}_{ij} = 0$. We can compute the necessary matrices
in $\bigO(n^{\omega}\log k)$ time, where $\omega < 2.373$ is the matrix multiplication exponent~\cite{alman2021refined}. Once $i$ is found, we use Dijkstra's algorithm to recover the path in $\bigO(n^2)$ time.

Overall we have a worst-case time complexity of $\bigO(n^{2+\omega} \log k)$.

Note that in practice, we do not use the matrix multiplication method. Instead,
we compute a shortest path tree from every node. This leads to a slower theoretical time but the algorithms
are still practical as demonstrated in the experiments.

\subsection{Practical improvements}
\label{sec:speeding}
We discuss several modifications, which speed up \algbac{}, and/or might improve the solution quality in practice. 

\ptitle{Algorithm \algbacb{}} The first modification to \algbac{}, given in Algorithm~\ref{algo:bacb}, is
aimed at improving solution quality, at the expense of a slightly larger
running time. We will call this modified algorithm \algbacb{}. Note that the $\frac{1}{6k}$-approximation guarantee comes from the fact
that we add an edge $ij$ to $G$ whenever $d(i, j) < \frac{R}{2k}$. If we can increase
this cutoff to, say $R \times \alpha$, \emph{and} still find a feasible set, then the found set $S$ 
is guaranteed to have $\diver{S} \geq R\alpha$, that is we will obtain an $\alpha$-approximation.

\begin{algorithm}[t]
\caption{\algbacb{U, d, k}, an $\frac{1}{6k}$-approx. algorithm for AMMD.}\label{algo:bacb} 
\begin{algorithmic}[1]
\Require space $(U, d)$ and integer parameter $k\geq 2$.
\State $R_1 < \ldots < R_{m} \define $ all unique positive distances sorted.
\For{every $i = 1, \ldots, m$} 
	\State ${U}' \define \algclust{U, d, R_i}$. 
	\If {$\algextract{U', d, R_i /(2k), k}$ exists}
		\State $a \define \min\set{s \mid R_i / (2k) < R_s}$, $b \define i$.
		\While {$a \leq b$}
			\State $t \define \floor{\frac{a+b}{2}}$.
			\If {$\algextract{U', d, R_t, k}$ exists}
				$a \define t + 1$
			\Else{}
				$b \define t - 1$
			\EndIf
		\EndWhile
	\EndIf
\EndFor 

\Ensure the set returned by $\algextract{}$ with the largest $\diver{}$ value.
\end{algorithmic}
\end{algorithm}

For every $R$ in the iteration of \algbacb{}
that gives a feasible solution for the threshold $\frac{R}{2k}$, we will try to
improve the solution value by searching for a cutoff value larger than
$\frac{R}{2k}$ when constructing the graph $G$ (line~\ref{line:create_G} of Algorithm~\ref{algo:extract}).
If \algbacb{} is unable to find a feasible solution for a
certain $R$, then we continue
iterating to the next $R$. 

To this end, if \algbacb{} has found a feasible set for some $R$, we use the binary search
to search for larger cutoffs in $[\frac{R}{2k},R]$.
Note that only when we use the cutoff $
\frac{R}{6k}$ are we theoretically guaranteed that we can extract $k$
independent points. Nonetheless, \algbacb{} will attempt to do this for larger
cutoffs as well. The binary search requires $\bigO(\log n)$ tests for
a single $R$ since we can assume that the cutoff is one of the distances. Hence,
the computational complexity of \algbacb{} is in $\bigO(n^{2+\omega} \log k \log
n)$.

\ptitle{Algorithm \algbacf{}} This algorithm speeds up \algbac{} while at the same time attempting to improve solution quality. Similarly to Section~\ref{sec:approxnk2}, we can replace the loop of Algorithm~\ref{algo:bac} (line~\ref{line:loop}) with a binary search in an attempt to find a maximal $R$ for which we find a feasible solution. This reduces the iterations from $\bigO(n^2)$ to $\bigO(\log n)$.

\begin{algorithm}[t]
\caption{\algbacf{U, d, k}, an $\frac{1}{6k}$-approx. algorithm for AMMD.}\label{algo:bacf} 
\begin{algorithmic}[1]
\Require space $(U, d)$ and integer parameter $k\geq 2$.
\State $R_1 < \ldots < R_{m} \define $ all unique positive distances sorted.
\State $a \define 1$, $b \define m$
\While {$a \leq b$}
    \State $t \define \floor{\frac{a+b}{2}}$.
	\State ${U}' \define \algclust{U, d, R_t}$. 
	\If {$\algextract{U', d, R_t / (2k), k}$ exists}
		$a \define t + 1$
	\Else{}
		$b \define t - 1$
	\EndIf
\EndWhile
\State $i \define a-1$.
\State ${U}' \define \algclust{U, d, R_i}$. 

\State $a \define \min\set{s \mid R_i / (2k) < R_s}$, $b \define i$.
\While {$a \leq b$}
    \State $t \define \floor{\frac{a+b}{2}}$.
	\If {$\algextract{U', d, R_t, k}$ exists}
		$a \define t + 1$
	\Else{}
		$b \define t - 1$
	\EndIf
\EndWhile

\Ensure the set returned by $\algextract{}$ with the largest $\diver{}$ value.
\end{algorithmic}
\end{algorithm}

Note that unlike in Section~\ref{sec:approxnk2}, this binary search might not find the globally largest $R$ value for which we can extract a feasible solution.
This is because there may be distance values $R_i > R_j > R'$ such that $R_i$ yields a feasible solution while $R_j$ does not.

However, we can still sort the
unique distances and use binary search to find
$R_{j}$ with a feasible solution, say $S_j$, such that the next value
$R_{j+1} > R_{j}$ does not yield a feasible solution.

Moreover, any $R \leq R'$ yields a feasible solution. This implies $R_{j} \geq R'$ and we get the same guarantee as \algbac{}, because
\[
	\diver{S_{j}} \geq \frac{R_j}{2k} \geq \frac{R'}{2k} \geq \frac{R^*}{6k}.
\]

Similarly to \algbacb{}, \algbacf{} then attempts to improve the cutoff
value for constructing $G$ by again binary searching for an improved cutoff
in the interval $[\frac{R_j}{2k},R_j]$ for which \algextract{} still finds a feasible
solution.
In the worst case, this adds
another $\bigO(\log n)$ iterations, which does not change the asymptotic
running time of the algorithm.
The running time for solving $\bigO(\log n)$ MA instances, as given by Proposition~\ref{prop:multima},
is dominated by the time needed to check for shortest paths of length $2k-1$.
In summary, the running time of \algbacf{} is
\[
	\bigO(n^2 \log n + 2 n^{\omega} \log n \log k) = \bigO(n^{\omega} \log n \log k).
\]

This is a considerable improvement over the running time of
\algbac{}, making \algbacf{} orders of magnitude faster while retaining the same theoretical guarantees.

\ptitle{Further improvements}
As the graph $G$ constructed in Algorithm~\ref{algo:extract} may be split into multiple disconnected components, we improve the search for independent sets by looking for the cycles, antichains, and long shortest paths in each weakly connected component of $G$ separately. We then take the union of the independent sets for each component, aiming to have $k$ points in total.

In addition, rather than choosing the centers arbitrarily in Algorithm~\ref{algo:ballcover}, we start by picking one of the vertices with the largest $\dmin$ distance. This heuristic is similar to the approach for solving MMD~\citep{ravi1994heuristic}. For the subsequent iterations, we choose the point furthest from the current set of chosen centers, as in Algorithm~\ref{algo:greedydmin}.

Finally, in practice many of the unique distances $R_1, \ldots, R_m$ may result in the same clustering ${U}'$ in Algorithms~\ref{algo:bac} and~\ref{algo:bacb}. We avoid these duplicate computations by grouping the $R$ values that yield the same clustering.

\section{Experiments}
\label{sec:exps}


\begin{table}[t]
  \caption{Characteristics of the datasets: size, smallest and largest distances, and number of unique distances. The total number of pairwise distances is $|U|(|U|-1)$.}
  \label{tab:datasets}
  \begin{tabular}{@{}lrrrrr@{}}
    \toprule
    Data  & type & $|U|$ & $\min d$ & $\max d$ & dists \\
    \midrule
    \emph{ft70} \cite{tsplib}  & infra & 70 & 331 & 2\,588 & 1\,441 \\
    \emph{kro124} \cite{tsplib} & infra & 100 & 81 & 4\,309 & 3\,297 \\
    \emph{celegans} \cite{kunegis2013konect} & bio & 297 & 1 & 24 & 24 \\
    \emph{rbg323} \cite{tsplib} & infra & 323 & 0 & 21 & 23 \\
    \emph{wiki-vote} \cite{snapnets} & social & 1\,300 & 1 & 9 & 9 \\
    \emph{US airports} \cite{toreopsal,kunegis2013konect} & flight & 1\,402 & 1 & 169\,685 & 27\,237 \\
    \emph{moreno} \cite{moody2001peer,kunegis2013konect} & social & 2\,155 & 1 & 52 & 52 \\
    \emph{openflights} \cite{toreopsal,kunegis2013konect} & flight & 2\,868 & 1 & 17 & 17 \\
    \emph{cora} \cite{vsubelj2013model,kunegis2013konect} & citation & 3\,991 & 1 & 45 & 45 \\
    \emph{bitcoin} \cite{kumar2016edge,kumar2018rev2} & trust & 4\,709 & 0 & 134 & 130 \\
    \emph{gnutella} \cite{snapnets} & P2P & 5\,153 & 1 & 21 & 21 \\
  \bottomrule
\end{tabular}
\end{table}

\begin{table}[] \centering
\caption{Performance comparison between our approximation algorithms (\textbf{\algbac{}}, \textbf{\algbacb{}}, and \textbf{\algbacf{}}), the greedy algorithm (\textbf{Greedy}), and the random baseline (\textbf{Rand.}), relative to the optimum solution value (Opt.) for fixed parameter value $k=10$.}
\label{tab:performsmallk}
\tcbset{colframe=black!5, colback=black!5, size=fbox, on line}
\setlength{\tabcolsep}{0pt}
\begin{tabular*}{\columnwidth}{@{\extracolsep{\fill}}lrrrrrr@{}}
\toprule
Data & \algbac{} & \algbacb{} & \algbacf{} & Greedy & Rand. & Opt. \\
\midrule
\emph{ft70} & 95\% & \tcbox{98\%} & 95\% & 95\% & 63\% & 786 \\
\emph{kro124p} & 88\% & \tcbox{89\%} & 88\% & 88\% & 45\% & 1 136 \\
\emph{celegans} & \tcbox{100\%} & \tcbox{100\%} & \tcbox{100\%} & \tcbox{100\%} & 29\% & 7 \\
\emph{rbg323} & 80\% & \tcbox{100\%} & 80\% & 80\% & 0\% & 15 \\
\emph{wiki-vote} & \tcbox{80\%} & \tcbox{80\%} & \tcbox{80\%} & \tcbox{80\%} & 40\% & 5 \\
\emph{US airports} & \tcbox{100\%} & \tcbox{100\%} & \tcbox{100\%} & \tcbox{100\%} & 0\% & 64 036 \\
\emph{moreno} & \tcbox{100\%} & \tcbox{100\%} & \tcbox{100\%} & \tcbox{100\%} & 27\% & 22 \\
\emph{openflights} & \tcbox{100\%} & \tcbox{100\%} & \tcbox{100\%} & \tcbox{100\%} & 33\% & 9 \\
\emph{cora} & \tcbox{100\%} & \tcbox{100\%} & \tcbox{100\%} & \tcbox{100\%} & 19\% & 27 \\
\emph{bitcoin} & \tcbox{100\%} & \tcbox{100\%} & \tcbox{100\%} & \tcbox{100\%} & 29\% & 77 \\
\emph{gnutella} & 85\% & \tcbox{92\%} & 85\% & 85\% & 38\% & 13 \\
\midrule
Average & 93.4\% & \tcbox{96.3\%} & 93.4\% & 93.4\% & 29.3\% & 100\% \\
\bottomrule
\end{tabular*}
\end{table}

\begin{figure*}[th!]

\centering
\subfloat[\emph{ft70} dataset]{\includegraphics{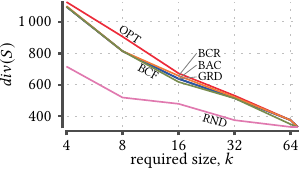}}
\subfloat[\emph{kro124} dataset]{\includegraphics{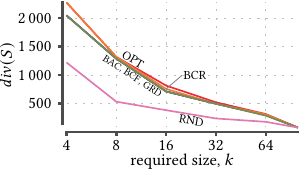}}
\subfloat[\emph{rbg323} dataset]{\includegraphics{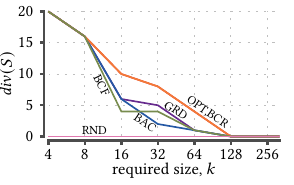}}
\caption{The y-axis shows the diversity score of the solutions provided by the algorithms. The x-axis shows the parameter $k$, which is the required solution size. All solutions converge to the same set as $k$ approaches $n$.  \label{fig:experimentsk}}
\end{figure*}

Next we describe our experiments.
All experiments were performed on an Intel\,Core\,i5-8265U processor at~1.6 GHz with 16\,GB\,RAM. 
Our methods were implemented in Python~3.8 and are publicly available.\footnote{\url{https://version.helsinki.fi/dacs/ammd}}

\subsection{Setup}

\ptitle{Data}
For data, we used weighted digraphs that we converted into an asymmetric distance space by computing the metric closure, that is, we define the distance $d(u, v)$ to be 
the weighted shortest path length between $u$ and $v$. We only used the largest strongly connected component.
Table~\ref{tab:datasets} shows the data and several statistics used for evaluation.
The column $\min d$ (resp. $\max d$) denotes the smallest (resp.
largest) distance present. The last column shows the
number of unique distances, which heavily
influences the running time of our algorithms from Section~\ref{sec:approx}.
The datasets \emph{ft70}, \emph{kro124} and \emph{rgb323} are asymmetric traveling salesman instances from the public library TSPLIB \cite{tsplib}. The \emph{bitcoin} dataset originally has edge weights between $[-10,10]$, which we have rescaled to nonnegative weights between $[0,20]$.

\ptitle{Baselines} Besides our proposed algorithms \algbac{},
\algbacb{}, and \algbacf{}, we have also implemented
and tested the following baselines: 

\pttitle{Greedy} Iteratively pick vertices maximizing the $\dmin$ distance towards the already chosen set of vertices as in Algorithm~\ref{algo:greedydmin}. To improve the performance, we start with a vertex of an edge with the maximum $\dmin$ distance.

\pttitle{Random} Select a subset $S \subseteq U$ of size $|S|=k$ uniformly at random. Repeat 10 times and return the set with the highest score $\diver{S}$.

\pttitle{Optimal} Computes an optimal solution by reducing it to solving multiple $k$-clique problems. First, sort all the unique weights in the AMMD instance.
The optimum $R^*$ is equal to one of the unique weights.
Let $R$ be one of the unique weights. Create an undirected graph $G_R = (U, E)$ with the same nodes as our AMMD instance, and $ij \in E$ if and only if both $d(i,j) \geq R$ and $d(j, i) \geq R$.
For every $R \leq R^*$ the graph $G_R$ will have a clique of size $k$, and for every $R>R^*$ the graph $G_R$ does not contain a clique of size $k$, by the optimality of $R^*$.
A binary search on the sorted list of unique weights finds the optimum $R^*$, by solving at most $\bigO(\log n)$ $k$-clique problems.
Note that the $k$-clique problem can be solved optimally in $\bigO(n^{\omega k/3})$ time \cite{nevsetvril1985complexity}. 
A similar idea has been used to optimally solve symmetric MMD and related fairness variants \cite{akagi2018exact,wang2023max}.

\subsection{The diversity scores of the returned sets}
\label{exp:per}

\ptitle{Small $k$ regime}
As mentioned in the introduction, in a typical AMMD setting one is interested in finding solutions of small size.
We compare the performance of our algorithms \algbac{}, \algbacb{} and \algbacf{} in Table~\ref{tab:performsmallk} with the aforementioned baselines Greedy, Random, and Optimal on all the datasets from Table~\ref{tab:datasets}, for the choice of $k=10$. 

Observe that \algbacb{} achieves the highest score in every experiment and finds the optimal value on seven datasets. In addition, \algbacf{} achieves the same solution quality as \algbac{} and Greedy while being significantly faster than \algbac{} and having the approximation guarantee discussed in Section~\ref{sec:speeding}. We also note that Random often performs very poorly, even with an increased number of repetitions. 

\algbacf{}, Greedy, and Random were very fast and took less than three seconds to run, whereas \algbacb{} took several minutes on some datasets. On \emph{wiki-vote} and \emph{gnutella} datasets, our baseline Optimal exceeded a time limit of two hours, but we were able to find the optimal values by manually checking that no cliques of size $k$ exist when only edges with higher weight are considered.  The running times on the real-world datasets are shown in Table~\ref{tab:running_times} in the Appendix.

\ptitle{Performance for any $k$}
We were able to run all the algorithms for all values of $k$ on the three small traveling salesman instances of Table~\ref{tab:datasets} within a reasonable time.
Figure~\ref{fig:experimentsk} shows the results, which are consistent with our previous findings. \algbacb{} performs very well also for larger values of $k$ while \algbac{}, \algbacf{}, and Greedy are often slightly worse. Note that as $k$ increases, the output of all algorithms converges to the same set, which is the entire universe of points.

\subsection{Running times of the algorithms}
\label{exp:rt}
\begin{figure}[ht!]
\centering
\subfloat[Synthetic scale-free graphs.\label{fig:runtimesa}]{\includegraphics{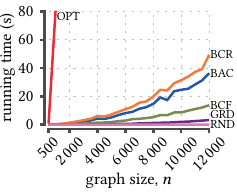}}
\hfill
\subfloat[Synthetic complete digraph.\label{fig:runtimesb}]{\includegraphics{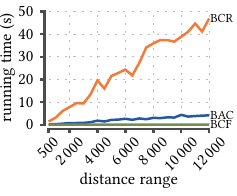}}
\caption{Running time of our algorithms as a function of the generated graph size and the size of the interval from which distances are sampled.
\label{fig:runtimes}}
\end{figure}

In this section, we compare the running time of our algorithms \algbac{}, \algbacb{}, and \algbacf{} on synthetically generated graph instances.

\ptitle{Increasing graph size}
In the first experiment, we generated directed (unweighted) scale-free networks of varying size
$n$ according to the model by~\citet{bollobas2003directed}. These synthetic
graphs are weakly connected, and we make them strongly connected by adding
directed edges in both directions. We set $k = 10$, and defined distances as the shortest path distances between pairs of
nodes (metric closure). The diameter of scale-free networks typically scales
about logarithmically with graph size~\cite{bollobas2004diameter}, so the
number of unique distances in our AMMD instances does not change that drastically when
increasing the graph size $n$. For example, when generating instances of size
$n=100$ and $n=3\,200$, the diameter roughly only doubled from 5 to 10.
Figure~\ref{fig:runtimesa} shows the results of the average running time over
several repeats. It shows that although the number of unique distances is
relatively low, the difference in speed between \algbacb{} and \algbacf{} is
still substantial. For $n=12\,000$ the running time of \algbacb{} was about
49s while \algbacf{} required less than 14s. As expected, the running time of Optimal scales exponentially while Greedy and Random are very fast.

\ptitle{Unique distances}
In a second timing experiment, we kept both $n$ and $k$ fixed and increased the number of unique distances.
We generated complete weighted digraphs of size $n=400$, where the weights are random
positive integers below a variable upper bound. To guarantee the triangle
inequality, we again took the metric closure. We set $k=10$. 
Figure~\ref{fig:runtimesb} shows the results. The theoretical speed-up of
\algbacf{} when compared to \algbac{} and \algbacb{}, as discussed in
Section~\ref{sec:speeding}, is clearly visible in our implementations on
practical data as well. Note that the running times of Greedy and Random do not depend on the number of unique distances, while Optimal takes prohibitively long for these instances. 

\section{Concluding remarks}\label{sec:conclusions}
This paper initiated the study on the asymmetric generalization of the Max-Min
Diversification problem, denoted as AMMD, which does not appear to have been
studied before. We provided an approximation algorithm with an approximation
factor of $\frac{1}{6k}$ and running time equivalent to matrix multiplication
with logarithmic factors. In practice, \algbacb{} outperformed all the baselines, returning optimal result in most of the cases,
while having pratical running times, and having theoretical guarantees.

Regarding the hardness of approximation, we leave it
as an open question whether one can improve the $\frac{1}{2}+\epsilon$
inapproximability result that follows from MMD, which is the restriction of
AMMD to symmetric distances. There is a large gap between our approximation
ratio and the inapproximability bound, and it would be interesting for future
work to narrow down the gap.
For example, by designing a constant-factor approximation algorithm for AMMD or showing that no such algorithm exists under reasonable assumptions.

\begin{acks}
This research is supported by the \grantsponsor{⟨malsome⟩}{Academy of Finland project MALSOME}{} (\grantnum[]{malsome}{343045}) and by the \grantsponsor{⟨hiit⟩}{Helsinki Institute for Information Technology (HIIT)}{}.
\end{acks}

\bibliographystyle{ACM-Reference-Format}
\balance
\bibliography{maxmin_kdd}


\clearpage 
\appendix
\section{Proofs}
\subsection{Proof of Theorem~\ref{thm:epsisymmgreedy}}
We start by proving the approximation ratio of Algorithm~\ref{algo:greedydmin}.
For a point $u \in U$ and $r>0$, let $B(u,r) = \{v \in U | \dmin(v,u) < r\}$ denote the $\dmin$-ball around $u$ with radius $r$.
We first prove the Property~\ref{prop:ballsgreedy}.
\begin{prop}
\label{prop:ballsgreedy}
If a pseudometric space $(U,d)$ is $\epsilon$-symmetric, then for any $x,y \in B(z,r)$ it holds that $\dmin(x,y) < (2+\epsilon)r$.
\end{prop}
\begin{proof}
It holds that
\begin{align*}
\dmin(x,y) &= \min\{d(x,y),d(y,x)\} \\
&\leq \min\{d(x,z)+d(z,y),d(y,z)+d(z,x)\}.
\end{align*}
If $d(x,z) \leq d(z,x)$, then 
\begin{align*}
\dmin(x,y) \leq d(x,z)+d(z,y) < r + (1+\epsilon)r = (2+\epsilon)r.
\end{align*}
Otherwise $d(x,z) > d(z,x)$, and then
\begin{align*}
\dmin(x,y) \leq d(y,z)+d(z,x) < (1+\epsilon)r + r = (2+\epsilon)r,
\end{align*}
and Property~\ref{prop:ballsgreedy} follows.
\end{proof}
Now we continue with the proof of Theorem~\ref{thm:epsisymmgreedy}.
Consider an optimum solution $O=\{o_1,\ldots,o_k\}$ of $k$ distinct points, with optimal value $\diver{O } = R^*$.
Property~\ref{prop:ballsgreedy} implies that the $k$ balls $B_i = B(o_i,\frac{R^*}{2+\epsilon})$ for each $i=1,\ldots,k$ are pairwise disjoint. Indeed, assume two balls $B_i$ and $B_j$ share a common point $u$, then because $o_i \in B(u,\frac{R^*}{2+\epsilon})$ and $o_j \in B(u,\frac{R^*}{2+\epsilon})$, Property~\ref{prop:ballsgreedy} states that $\dmin(o_i,o_j)<R^*$, contradicting the optimality of $O$.

While $|S|<k$ in Algorithm~\ref{algo:greedydmin}, there exists some ball $B_i$ that does not contain any point from $S$. Therefore, $o_i$ is available for selection and $\dmin(o_i,S) \geq \frac{R^*}{2+\epsilon}$. As Algorithm~\ref{algo:greedydmin} picks the vertex with the maximum $\dmin$ distance towards the $S$, it holds that $\diver{S}\geq \frac{R^*}{2+\epsilon}$ in each iteration. The performance guarantee follows. 

To show that this ratio is tight, consider the following example of a pseudometric on three points $\{x,y,z\}$ with $d(x,y) = d(y,x) = 1$,
$d(x,z) = d(y,z) = \frac{1}{2+\epsilon}$ and 
$d(z,x) = d(z,y) = \frac{1+\epsilon}{2+\epsilon}$. This instance is $\epsilon$-symmetric and satisfies the directed triangle inequality. For $k=2$, the optimum is $O=\{x,y\}$ with $R^*=1$, but Algorithm~\ref{algo:greedydmin} returns a set with value $\frac{1}{2+\epsilon}$ if the first point chosen is $z$.

Algorithm~\ref{algo:greedydmin} can be implemented to run in $\bigO(kn)$ time by maintaining $d(u,S)$ for each $u \in U$ in each iteration. If $v$ is selected in an iteration, then the distances are updated as $d(u,S \cup \{v\}) = \min\{d(u,S),d(u,v)\}$. There are $k$ iterations and $n$ updates in each iteration, with each update requiring constant time.

\subsection{Proof of Theorem~\ref{thm:naivemaxanti}}
Recall that $R^*$ is the optimal value for an AMMD instance space $(U,d)$ with parameter $k$, and $O$ is a corresponding optimal set of $k$ points.
There are at most $n(n-1)$ unique pairwise distances in $d$, and $R^*$ equals one of them. Algorithm~\ref{algo:app_ma} will iterate over all unique distances $R$ and creates a digraph $G_R = (U,A)$ with $ij \in E$ if and only if $d(i,j)< \frac{R}{n-k+1}$. 

Now we argue that for any $R \leq R^*$, and for any two $x, y \in O, x \neq y$, it must hold that $y$ is unreachable from $x$ in $G_R$.
Suppose that $y$ is reachable from $x$. So there is a path $p$ in $G_R$ from $x$ to $y$.
We may assume that $p$ only has vertices that are not in $O$, except for the first vertex $x$ and the last vertex $y$. Otherwise, we continue the argument by shortening $p$ to start at $x$ until encountering the first point in $O$. Using the triangle inequality along the edges in $p$, it follows that
\[
	d(x,y) < (n-k+1) \frac{R}{n-k+1} = R,
\]
as there are at most $n-k+1$ edges in $p$, and the distance between consecutive points in $p$ is strictly smaller than $\frac{R}{n-k+1}$ by definition of $G_R$. This contradicts with $d(x,y) \geq R^*$, by optimality of $x, y \in O$.

As the points in an optimal solution are pairwise unreachable in $G_R$ for $R \leq R^*$, it follows that the maximum antichain $M$ in $G_R$ will have at least $k$ vertices, and we simply select $k$ arbitrary vertices from $M$ and return them as our final solution. As two distinct vertices $i, j \in M, i \neq j$ satisfy $d(i,j) \geq \frac{R}{n-k+1}$, the approximation guarantee follows by iterating over all unique distances so that for one of the iterations we have $R=R^*$.
The running time follows from Proposition~\ref{caceres}, by using the maximum antichain algorithm from \cite{caceres2022minimum}.

\subsection{Table of running times}

\begin{table}[H] \centering
\caption{Running times for our approximation algorithms (\textbf{\algbac{}}, \textbf{\algbacb{}}, and \textbf{\algbacf{}}), the greedy algorithm (\textbf{Greedy}), and the random baseline (\textbf{Rand.}) for fixed parameter value $k=10$ on the real-world datasets from Table~\ref{tab:datasets}.}
\tcbset{colframe=black!5, colback=black!5, size=fbox, on line}
\label{tab:running_times}
\begin{tabular}{@{}lrrrrr@{}}\toprule
Data & BAC & BCR & BCF & Greedy & Rand. \\
\midrule
\emph{ft70} & 1s & 4s & 0s & 0s & 0s \\
\emph{kro124p} & 1s & 5s & 0s & 0s & 0s \\
\emph{celegans} & 0s & 0s & 0s & 0s & 0s \\
\emph{rbg323} & 0s & 1s & 0s & 0s & 0s \\
\emph{wiki-vote} & 0s & 1s & 0s & 0s & 0s \\
\emph{US airports} & 1m 44s & 7m 25s & 0s & 0s & 0s \\
\emph{moreno} & 3s & 15s & 0s & 0s & 0s \\
\emph{openflights} & 2s & 2s & 1s & 0s & 0s \\
\emph{cora} & 9s & 41s & 1s & 0s & 0s \\
\emph{bitcoin} & 51s & 5m 13s & 2s & 0s & 0s \\
\emph{gnutella} & 11s & 23s & 2s & 0s & 0s \\
\bottomrule
\end{tabular}
\end{table}

\end{document}